\documentclass[12pt,draftcls,onecolumn]{IEEEtran}
\usepackage{subfigure,cite,graphicx,amsmath,amssymb,eufrak,mathrsfs,epstopdf,epsf,epsfig,psfrag}
\begin{document}

\title{Capacity of Strong and Very Strong Gaussian Interference Relay-without-delay Channels}
%\date{}                                           % Activate to display a given date or no date
\author{
\authorblockN{Hyunseok Chang and Sae-Young Chung}\\
\authorblockA{Department of EE, KAIST, Daejeon, Korea \\
Email: hyunseok.chang@kaist.ac.kr, sychung@ee.kaist.ac.kr }
}
\maketitle

%%%%%%%%%%%%%%%%%%%%%%%%%%%%%%%%%%%%%%%%%%%%%%%%%%%%%%%%%%%%%%%%%%%%%%%%%%%%%%%%%%%%%%%%%%%%%%%%%%%%%%%%%%%%%%%%%%%%%%%%%%%%%%%%%%%%%%%%%%%%%%%%%%%%%

\newtheorem{definition}{Definition}%[section]
\newtheorem{theorem}{Theorem}%[section]
\newtheorem{lemma}{Lemma}%[section]
\newtheorem{example}{Example}
\newtheorem{corollary}{Corollary}%[section]
\newtheorem{proposition}{Proposition}%[section]
\newtheorem{conjecture}{Conjecture}%[section]
\newtheorem{remark}{Remark}%[section]

\def \diag{\operatornamewithlimits{diag}}
\def \min{\operatornamewithlimits{min}}
\def \max{\operatornamewithlimits{max}}
\def \log{\operatorname{log}}
\def \max{\operatorname{max}}
\def \rank{\operatorname{rank}}
\def \out{\operatorname{out}}
\def \exp{\operatorname{exp}}
\def \arg{\operatorname{arg}}
\def \E{\operatorname{E}}
\def \tr{\operatorname{tr}}
\def \SNR{\operatorname{SNR}}
\def \SINR{\operatorname{SINR}}
\def \dB{\operatorname{dB}}
\def \ln{\operatorname{ln}}
\def \th{\operatorname{th}}

%\newpage

\begin{abstract}
%        Relaying is regarded as one of the most important strategies in wireless communications to increase throughput and cell coverage.
%        Because of its practical importance, it has received a great deal of attention not only from the academy but also from the industry.
%        For the single source-destination relay channel, several relaying schemes are proposed and shown to be optimal for some classes of the channel.
%        Recently, relay-without-delay channel received attention and it is shown that instantaneous amplify-and-forward relaying can achieve the capacity for some classes of Gaussian relay-without-delay channel.
%        For the relay networks with multiple source-destination pairs, however, capacity characteristic is relatively less known because of the interference.
%        In this context, we consider interference relay-without-delay channel and layered interference relay network among relay networks with multiple source-destination pairs.

        In this paper, we study the interference relay-without-delay channel which is an interference channel with a relay helping the communication. We assume the relay's transmit symbol depends not only on its past received symbols but also on its current received symbol, which is an appropriate model for studying amplify-and-forward type relaying when the overall delay spread is much smaller than the inverse of the bandwidth.
        For the discrete memoryless interference relay-without-delay channel, we show an outer bound using genie-aided outer bounding.
        For the Gaussian interference relay-without-delay channel, we define strong and very strong interference relay-without-delay channels and propose an achievable scheme based on instantaneous amplify-and-forward (AF) relaying.
        We also propose two outer bounds for the strong and very strong cases. Using the proposed achievable scheme and outer bounds, we show that our scheme can achieve the capacity exactly when the relay's transmit power is greater than a certain threshold. This is surprising since the conventional AF relaying is usually only asymptotically optimal, not exactly optimal.
        The proposed scheme can be useful in many practical scenarios due to its optimality as well as its simplicity.

\end{abstract}


\section{Introduction}

 The performance of wireless communication systems is significantly affected by the interference since resources such as time, frequency, and space are often shared.
 The two-user interference channel, which is the smallest multiple source-destination pair communication channel with interference, has received a great deal of attention.
 However, single-letter capacity expressions for discrete memoryless and Gaussian interference channels are still unknown in general and known only for some limited cases.
 For strong \cite{Sato81,Costa87} and very strong \cite{Carleial78} interference channels, the capacity region is known in general.
 Recently, the sum-rate capacity for the very weak IC \cite{Shang09} has been discovered based on genie-aided outer bounding.
 Motivated by the sum-rate capacity for the very weak IC, the capacity region for general Gaussian interference channel is characterized to within one bit in~\cite{Etkin08}.

 Recently, many people have focused on the interference relay channel, which is an interference channel with a relay helping the communication.
 For the interference relay channel where the relay's transmit signal can be heard from only one source, the capacity is known under a certain condition using an interference forwarding scheme \cite{Maric08}.
 Rate splitting and decode-and-forward relaying approach is considered in \cite{Sahin07}, showing that transmitting the common message only at the relay achieves the maximum achievable sum rate for symmetric Gaussian interference relay channels.
 For a Gaussian interference channel with a cognitive relay having access to the messages transmitted by both sources, \cite{Sridharan08} proposed a new achievable rate region.

 In network information theory, it is typically assumed that the relay's transmit symbol depends only on its past received symbols. However, sometimes it makes more sense to assume that the relay's current transmit symbol depends also on its current received symbol, which is a better model for studying AF type relaying if the overall delay spread including the path through the relay is much smaller than the inverse of the bandwidth.
 This channel has recently received attention~\cite{Elgamal07,Nam11}.
 Relay-without-delay channel is a single source-destination pair communication system with a relay helping the communication such that it encodes its transmit symbol not only based on its past received symbols but also on the current received symbol.
 For the Gaussian relay-without-delay channel, instantaneous amplify-and-forward relaying achieves the capacity if the relay's transmit power is greater than a certain threshold~\cite{Elgamal07}.

 In this paper, we consider interference relay-without-delay channel where two source nodes want to transmit messages to their respective destination nodes and a relay without delay helps their communication.
 We study both discrete memoryless and Gaussian memoryless interference relay-without-delay channels.
 We present an outer bound using genie-aided bounding for the discrete memoryless interference relay-without-delay.
 For the Gaussian interference relay-without-delay channel, we define strong and very strong Gaussian interference relay channels motivated by the Gaussian strong and very interference channels.
 For strong and very strong Gaussian interference relay-without-delay channels, we show new outer bounds using genie-aided outer bounding such that both receivers know the received sequence at the relay.
 We also propose an achievable scheme using Gaussian codebooks, simultaneous non-unique decoding, and instantaneous amplify-and-forward relaying.
 The complexity of the proposed relaying is very low compared to other more complicated relaying schemes including decode-and-forward and compress-and-forward since it only has symbolwise operations. Despite its simplicity, we show that it can achieve the capacity exactly when the relay's transmit power is greater than a certain threshold for the very strong case. For the strong case, the same is true if some additional conditions are satisfied.
 This is surprising since it means that such a simple symbolwise relaying can be optimal in a complicated communication scenario where there are two source-destination pairs interfering with each other.
 The proposed scheme would be practically useful since it can be optimal and at the same time it is very simple.

This paper is organized as follows: In Section 2, we describe the discrete memoryless interference relay-without-delay channel model and derive an outer bound.
In Section 3, we define strong and very strong Gaussian interference relay-without-delay channel and propose outer bounds for each case.
We propose an achievable scheme based on instantaneous amplify-and-forward relaying and show that this achievable scheme achieves the capacity for strong and very strong Gaussian interference relay-without-delay channels under certain conditions.
To show the existence of strong and very strong Gaussian interference relay-without-delay channels, we provide some examples.

\section{Discrete memoryless Interference Relay-without-delay Channel}
The discrete memoryless interference relay-without-delay channel consists of two senders $X_1\in \mathcal{X}_1,X_2\in \mathcal{X}_2$, two receivers $Y_1\in \mathcal{Y}_1,Y_2\in \mathcal{Y}_2$, and a set of conditional probability mass functions
$p(y_R|x_1,x_2)p\left(y_1,y_2|x_1,x_2,x_R,y_R\right)$ on $\mathcal{Y}_1 \times \mathcal{Y}_2 \times \mathcal{Y}_R$, one for each $(x_1,x_2,x_R,y_R) \in \mathcal{X}_1 \times \mathcal{X}_2 \times \mathcal{X}_R \times \mathcal{Y}_R$. It is denoted by $(\mathcal{X}_1 \times \mathcal{X}_2 \times \mathcal{X}_R, p(y_R|x_1,x_2)p(y_1,y_2|x_1,x_2,x_R,y_R),\mathcal{Y}_1 \times \mathcal{Y}_2 \times \mathcal{Y}_{R})$.
Then, the interference relay channel can be depicted as in Figure \ref{Figure 3.1}.

\begin{figure}[tch]
    \centerline{\includegraphics[width=15cm]{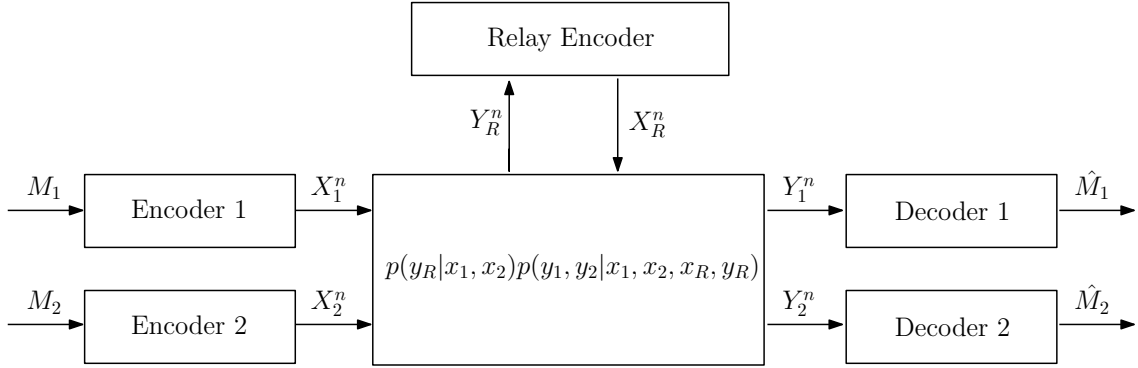}}
    \caption{ Discrete Memoryless Interference Relay-without-delay Channel.
    } \label{Figure 3.1}
\end{figure}

A $(2^{nR_1},2^{nR_2},n)$ code for the discrete memoryless interference relay-without-delay channel consists of two message sets $[1:2^{nR_1}]\triangleq \{1,2,\ldots,2^{nR_1}\}$ and $[1:2^{nR_2}]$, two encoders that assign codewords $x_1^n(m_1)$ and $x_2^n(m_2)$ to each messages $m_1 \in [1:2^{nR_1}]$ and $m_2 \in [1:2^{nR_2}]$, respectively.
The encoding function of a relay $f^n$ is defined as $x_{R,i}=f_i(y_R^{i})$, for $1 \leq i \leq n$.
For decoding, decoder i assigns a message $\hat{m}_i$ or an error $e$ to each received sequence $y_i^n \in \mathcal{Y}_i^n$, for $i \in \{1,2\}$.

The average probability of error is defined as $P_e^{(n)}=P\left((\hat{M}_1,\hat{M}_2) \neq (M_1,M_2)\right)$.
A rate pair $(R_1,R_2)$ is said to be achievable if there exists a sequence of $(2^{nR_1},2^{nR_2},n)$ codes such that $\lim_{n\rightarrow\infty}P_e^{(n)}=0$.
The capacity region of the discrete memoryless interference relay-without-delay channel is defined as the closure of the set of achievable rate pairs $(R_1,R_2)$.

For the discrete memoryless interference relay-without-delay channel, we get the following outer bound on the capacity region.
\begin{theorem}
For the discrete memoryless interference relay-without-delay channel, the capacity region is contained in the set of rate pairs $(R_1,R_2)$ such that
\begin{align*}
R_1 &\leq I(X_{1};Y_{R}|X_{2},Q)+I(X_{1};Y_{1}|X_{2},Y_{R},X_{R},Q),\\
R_2 &\leq I(X_{2};Y_{R}|X_{1},Q)+I(X_{2};Y_{2}|X_{1},Y_{R},X_{R},Q),
\end{align*}
for some $p(q)p(x_1|q)p(x_2|q)p(y_{R}|x_1,x_2)p(x_{R}|y_{R},q)p(y_1,y_2|x_1,x_2,x_R,y_R)$.
\end{theorem}
\proof
The joint probability mass function (pmf) induced by the encoders, relay, and decoders can be written in the following form:
\begin{align*}
&p(m_1,m_2,x_1^n,x_2^n,x_{R}^n,y_1^n,y_2^n,y_{R}^n)\\
&=2^{-n(R_1+R_2)}p(x_1^n|m_1)p(x_2^n|m_2)\left(\prod_{i=1}^{n}p_{Y_{R}|X_1,X_2}(y_{Ri}|x_{1i}x_{2i})\right)\left(\prod_{i=1}^{n}p_{X_{R,i}|Y_{R}^i}(x_{R,i}|y_{R}^{i})\right)\\
&\left(\prod_{i=1}^{n}p_{Y_{1},Y_{2}|X_1,X_2,X_{R},Y_{R}}(y_{1i},y_{2i}|x_{1i},x_{2i},x_{Ri},y_{Ri})\right)
\end{align*}
From Fano's inequality \cite{text_Cover}, we get
\begin{align*}
&H(M_1|Y_1^n,Y_{R}^n,M_2)\leq H(M_1|Y_1^n)\leq H(M_1|\hat{M_1})\leq n\epsilon_n,\\
&H(M_2|Y_2^n,Y_{R}^n,M_1)\leq H(M_2|Y_2^n)\leq H(M_2|\hat{M_2})\leq n\epsilon_n.
\end{align*}
where $\epsilon_n$ tends to zero as $n \rightarrow \infty$. An upper bound on $nR_1$ can be expressed as follows:
\begin{align*}
&nR_1=H(M_1)=H(M_1|M_2)\\
&\overset{(a)}\leq I(M_1;Y_1^n,Y_{R}^n|M_2)+n\epsilon_n\\
&=\sum_{i=1}^nI(M_1;Y_{Ri},Y_{1,i}|M_2,Y_{R}^{i-1},Y_{1}^{i-1})+n\epsilon_n\\
&\overset{(b)}= \sum_{i=1}^nI(M_1,X_{1i};Y_{Ri},Y_{1,i}|M_2,X_{2i},Y_{R}^{i-1},Y_{1}^{i-1})+n\epsilon_n\\
&\leq \sum_{i=1}^nI(M_1,M_2,X_{1i},Y_{1}^{i-1};Y_{Ri},Y_{1,i}|X_{2i},Y_{R}^{i-1})+n\epsilon_n\\
&= \sum_{i=1}^nI(M_1,M_2,X_{1i},Y_{1}^{i-1};Y_{Ri}|X_{2i},Y_{R}^{i-1})+\sum_{i=1}^nI(M_1,M_2,X_{1i},Y_{1}^{i-1};Y_{1,i}|X_{2i},Y_{R}^{i})+n\epsilon_n\\
&\overset{(c)}\leq \sum_{i=1}^nI(M_1,M_2,X_{1i},Y_{1}^{i-1},Y_{R}^{i-1};Y_{Ri}|X_{2i})+\sum_{i=1}^nI(M_1,M_2,X_{1i},Y_{1}^{i-1},Y_{R}^{i-1};Y_{1,i}|X_{2i},Y_{Ri},X_{Ri})+n\epsilon_n\\
&\overset{(d)}=\sum_{i=1}^nI(X_{1i};Y_{Ri}|X_{2i})+\sum_{i=1}^nI(X_{1i};Y_{1,i}|X_{2i},Y_{Ri},X_{Ri})+n\epsilon_n\\
&=nI(X_{1Q};Y_{R,Q}|X_{2Q},Q)+nI(X_{1Q};Y_{1Q}|X_{2Q},Y_{RQ},X_{RQ},Q)+n\epsilon_n\\
&=nI(X_{1};Y_{R}|X_{2},Q)+nI(X_{1};Y_{1}|X_{2},Y_{R},X_{R},Q)+n\epsilon_n
\end{align*}
where (a) follows from Fano's inequality, (b) and (c) follow since $X_{ji}$ is a function of $M_j$ for $j \in \{1,2\}$, respectively and $X_{Ri}=f_i(Y_{R}^i)$, and (d) holds since $(M_1,M_2,Y_{R}^{i-1},Y_1^{i-1}) \rightarrow (X_{1i},X_{2i}) \rightarrow Y_{Ri}$ and $(M_1,M_2,Y_{1}^{i-1},Y_{R}^{i-1}) \rightarrow (X_{1i},X_{2i},X_{Ri},Y_{Ri}) \rightarrow Y_{1i}$. In the above, $Q$ is a time sharing random variable such that $Q \sim \operatorname{Unif}[1:n]$ which is independent of $(X_1^n,X_2^n,Y_1^n,Y_2^n,Y_R^n,X_R^n)$ and we define $X_1=X_{1Q}$, $X_2=X_{2Q}$, $Y_1=Y_{1Q}$, $Y_2=Y_{2Q}$, $Y_R=Y_{RQ}$, $X_R=X_{RQ}$.

Similarly, an upper bound on $R_2$ can be expressed as follows:
\begin{align*}
R_2 \leq I(X_{2};Y_{R}|X_{1},Q)+I(X_{2};Y_{2}|X_{1},Y_{R},X_{R},Q)
\end{align*}
\endproof

\section{Gaussian Interference Relay-Without-Delay Channel}

We consider the Gaussian interference relay-without-delay channel model depicted in Figure \ref{Figure 3.2}, which is a Gaussian version of the discrete memoryless interference relay-without-delay channel considered in Section 2.
Suppose the relay node is equipped with 3 antennas, one for receiving and two for transmitting signals. Receiving and transmitting antennas are assumed to be isolated so that they do not interfere with each other.
\begin{figure}[tch]
    \centerline{\includegraphics[width=12cm]{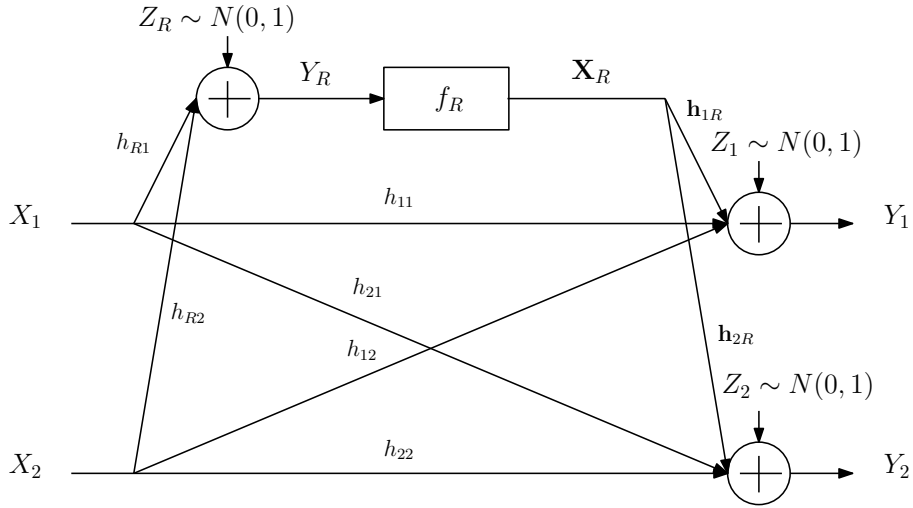}}
    \caption{ Gaussian Interference Relay-Without-Delay Channel.
    } \label{Figure 3.2}
\end{figure}
Then the relay's received signal $Y_{R}$ can be expressed as follows:
\begin{align*}
Y_{R}=h_{R1}X_1+h_{R2}X_2+Z_{R}
\end{align*}
where $h_{Rj}$ is the channel gain from source $j$ to the relay, $Z_{R}\sim N(0,1)$ is the zero-mean unit-variance Gaussian noise at the relay. Suppose an average power constraint $P$ on source nodes and $P_R$ on the relay node.
%Let the transmitting signal of the relay be denoted as $\textbf{X}_{R}$ and let the received signal of destination $i$ as $Y_i$.
The received signal of destination $i$, $Y_i$ can be expressed as follows:
\begin{align*}
Y_{1}=h_{11}X_1+h_{12}X_2+\textbf{h}_{1R}\textbf{X}_{R}+Z_1\\
Y_{2}=h_{21}X_1+h_{22}X_2+\textbf{h}_{2R}\textbf{X}_{R}+Z_2
\end{align*}
where $\textbf{h}_{iR}=[h_{iR}^{[1]},h_{iR}^{[2]}]$, $h_{iR}^{[k]}$ is the channel gain from the $k$-th antenna from the relay to destination $i$, $h_{ik}$ is the channel gain from source $k$ to destination $i$, for $i,k \in\{1,2\}$, and $Z_{i}\sim N(0,1)$ is the Gaussian noise at destination $i$.

\subsection{Very strong Gaussian interference relay-without-delay channel}
\begin{definition}
A Gaussian interference relay-without-delay channel is said to be \textbf{very strong Gaussian interference relay-without-delay channel} if
\begin{align*}
\left(h_{11}^2+h_{R1}^2\right)P&\leq\frac{\left(h_{21}h_{22}+h_{R1}h_{R2}\right)^2P}{\left(h_{22}^2+h_{R2}^2\right)^2P+h_{22}^2+h_{R2}^2}\\
\left(h_{22}^2+h_{R2}^2\right)P&\leq\frac{\left(h_{11}h_{12}+h_{R1}h_{R2}\right)^2P}{\left(h_{11}^2+h_{R1}^2\right)^2P+h_{11}^2+h_{R1}^2}
\end{align*}
\end{definition}
\vspace{0.15in}
\begin{theorem}
The capacity region of the very strong Gaussian interference relay-without-delay channel for
\begin{align}
\frac{(h_{2R}^{[1]}h_{22}h_{R1}-h_{1R}^{[1]}h_{11}h_{R2})^2+(h_{2R}^{[2]}h_{22}h_{R1}-h_{1R}^{[2]}h_{11}h_{R2})^2}{(h_{1R}^{[2]}h_{2R}^{[1]}-h_{1R}^{[1]}h_{2R}^{[2]})^2 h_{11}^2 h_{22}^2 } \leq \frac{P_R}{(h_{R1}^2+h_{R2}^2)P+1}\label{eq:pr}
\end{align}
is the set of rate pairs $(R_1,R_2)$ such that
\begin{align*}
R_1 \leq \frac{1}{2}\log\left(1+\left(h_{11}^2+h_{R1}^2\right)P\right)\\
R_2 \leq \frac{1}{2}\log\left(1+\left(h_{22}^2+h_{R2}^2\right)P\right)
\end{align*}
\end{theorem}

\proof

\textbf{Achievability.}
Achievability follows from the simultaneous non-unique decoding rule \cite{text_Elgamal} and instantaneous amplify-and-forward scheme at each relay.
If we apply instantaneous amplify-and-forward relaying, the transmit signal of relay $k$ can be expressed as follows:
\begin{align*}
\textbf{X}_{R}= \left[\begin{array}{ccc}
    \alpha_{R1} \\
    \alpha_{R2} \\
  \end{array}\right]Y_{R}
\end{align*}
where $\alpha_{Ri}$ is the amplifying factor of the $i$-th antenna of the relay.
Let sources 1 and 2 transmit codewords using independent Gaussian codebooks with power P. Then, the received power of the relay becomes $(h_{R1}^2+h_{R2}^2)P+1$ and the transmit power constraint for the relay can be expressed as follows:
 %Let the transmitting power of the relay node be denoted as $\textbf{P}_{R}$. Then $\textbf{P}_{R}$ can be expressed as follows:
%\begin{align*}
%\textbf{P}_{R}= \left[\begin{array}{ccc}
%    P_{R,1} \\
%    P_{R,2} \\
%  \end{array}\right]
%=\left[\begin{array}{ccc}
%    \alpha_{R,1}^2((h_{R,1}^2+h_{R,2}^2)P+1) \\
%    \alpha_{R,2}^2((h_{R,1}^2+h_{R,2}^2)P+1) \\
%  \end{array}\right]
%\end{align*}
%where $P_{R,i}$ is the power of the $i$-th antenna of the relay.
%Because of transmitting power constraint, transmitting power for each relay should not exceed $P$.
%Thus, the power constraint for the relay can be expressed as follow:
\begin{align*}
\alpha_{R1}^2+\alpha_{R2}^2 &\leq \frac{P_R}{(h_{R1}^2+h_{R2}^2)P+1}
\end{align*}

We set the amplify-and-forward factors of the relay satisfying the following equations.
\begin{align*}
&\textbf{h}_{1R}
\left[\begin{array}{ccc}
    \alpha_{R1} \\
    \alpha_{R2} \\
\end{array}\right]
=[h_{1R}^{[1]} \hspace{0.15 in} h_{1R}^{[2]}]
\left[\begin{array}{ccc}
    \alpha_{R1} \\
    \alpha_{R2} \\
\end{array}\right]
= \alpha_{R1}h_{1R}^{[1]} + \alpha_{R2}h_{1R}^{[2]} = \frac{h_{R1}}{h_{11}}\\
&\textbf{h}_{2R}
\left[\begin{array}{ccc}
    \alpha_{R1} \\
    \alpha_{R2} \\
\end{array}\right]
=[h_{2R}^{[1]} \hspace{0.15 in} h_{2R}^{[2]}]
\left[\begin{array}{ccc}
    \alpha_{R1} \\
    \alpha_{R2} \\
\end{array}\right]
=\alpha_{R1}h_{2R}^{[1]} + \alpha_{R2}h_{2R}^{[2]} = \frac{h_{R2}}{h_{22}}.
\end{align*}
%where $h_{j,R}^{[i]}$ is the channel gain from the $i$-th antenna of the relay to destination j.

By solving this for $\alpha_{R1}$ and $\alpha_{R2}$, we get:
\begin{align*}
\alpha_{R1}=\frac{h_{11}h_{R2}h_{1R}^{[2]}-h_{22}h_{R1}h_{2R}^{[2]}}{h_{11}h_{22}(h_{1R}^{[2]}h_{2R}^{[1]}-h_{1R}^{[1]}h_{2R}^{[2]})}\\
\alpha_{R2}=\frac{h_{22}h_{R1}h_{2R}^{[1]}-h_{11}h_{R2}h_{1R}^{[1]}}{h_{11}h_{22}(h_{1R}^{[2]}h_{2R}^{[1]}-h_{1R}^{[1]}h_{2R}^{[2]})}
\end{align*}

Using the above $\alpha_{R,1}$ and $\alpha_{R,2}$, the received signal of destination $i$ can be expressed as follows:
\begin{align*}
Y_{1}&=h_{11}X_1+h_{12}X_2+\textbf{h}_{1R}\textbf{X}_{R}+Z_1\\
&=\left(\frac{h_{11}^2+h_{R1}^2}{h_{11}}\right)X_1+\left(\frac{h_{11}h_{12}+h_{R1}h_{R2}}{h_{11}}\right)X_{2}+\left(\frac{h_{R1}}{h_{11}}\right)Z_{R}+Z_{1}\\
Y_{2}&=h_{21}X_1+h_{22}X_2+\textbf{h}_{2R}\textbf{X}_{R}+Z_2\\
&=\left(\frac{h_{21}h_{22}+h_{R1}h_{R2}}{h_{22}}\right)X_1+\left(\frac{h_{22}^2+h_{R2}^2}{h_{22}}\right)X_2+\left(\frac{h_{R2}}{h_{22}}\right)Z_{R}+Z_{2}
\end{align*}

Since $Y_1$ and $Y_2$ can be expressed in terms of $X_1$, $X_2$, and a scaled sum of independent Gaussian noises, we can equivalently regard this channel as the interference channel with Gaussian noise.
Therefore, using the simultaneous non-unique decoder, the achievable rate pairs $(R_1,R_2)$ can be expressed as follows:
\begin{align*}
&R_1 \leq \frac{1}{2}\log(1+(h_{11}^2+h_{R1}^2)P)\\
&R_2 \leq \frac{1}{2}\log(1+(h_{22}^2+h_{R2}^2)P)\\
&R_1+R_2 \leq \frac{1}{2}\log\left(1+\left(h_{11}^2+h_{R1}^2+\frac{(h_{11}h_{12}+h_{R1}h_{R2})^2}{h_{11}^2+h_{R1}^2}\right)P\right)\\
&R_1+R_2 \leq \frac{1}{2}\log\left(1+\left(h_{22}^2+h_{R2}^2+\frac{(h_{21}h_{22}+h_{R1}h_{R2})^2}{h_{22}^2+h_{R2}^2}\right)P\right)
\end{align*}

From the assumption of very strong Gaussian interference relay-without-delay channel, the constraints on $R_1+R_2$ are redundant since
\begin{align*}
R_1+R_2 &\leq \frac{1}{2}\log(1+(h_{11}^2+h_{R1}^2)P)+\frac{1}{2}\log(1+(h_{22}^2+h_{R2}^2)P) \\
&\leq \frac{1}{2}\log(1+(h_{11}^2+h_{R1}^2)P)+\frac{1}{2}\log\left(1+\frac{\left(h_{11}h_{12}+h_{R1}h_{R2}\right)^2P}{(h_{11}^2+h_{R1}^2)^2P+(h_{11}^2+h_{R1}^2)}\right) \\
&=\frac{1}{2}\log\left(1+\left(h_{11}^2+h_{R1}^2+\frac{\left(h_{11}h_{12}+h_{R1}h_{R2}\right)^2}{h_{11}^2+h_{R1}^2}\right)P\right)
\end{align*}

Thus, the achievable rate pairs $(R_1,R_2)$ can be expressed as follows:
\begin{align*}
&R_1 \leq \frac{1}{2}\log(1+(h_{11}^2+h_{R1}^2)P)\\
&R_2 \leq \frac{1}{2}\log(1+(h_{22}^2+h_{R2}^2)P)
\end{align*}

\textbf{Converse.}
Applying Theorem 1 to the Gaussian interference relay-without-delay channel, the upper bound on $R_1$ in Theorem 1 can be expressed as follows:
\begin{align*}
&R_1 \leq I(X_{1};Y_{R}|X_{2})+I(X_{1};Y_{1}|X_{2},Y_{R},\textbf{X}_{R})\\
&=h(Y_{R}|X_{2})-h(Y_{R}|X_1,X_2)+h(Y_1|X_2,Y_{R},\textbf{X}_{R})-h(Y_1|X_1,X_2,Y_{R},\textbf{X}_{R})\\
&\overset{(a)}=h(h_{R1}X_1+h_{R2}X_2+Z_{R}|X_{2})-h(Z_{R})+h(h_{11}X_1+h_{12}X_2+h_{1R}\textbf{X}_{R}+Z_1|X_2,Y_{R},\textbf{X}_{R})-h(Z_1)\\
&\overset{(b)}=h(h_{R1}X_1+Z_{R})-h(Z_{R})+h(h_{11}X_1+Z_1|h_{R1}X_1+Z_{R})-h(Z_1)\\
&=h(h_{11}X_1+Z_1,h_{R1}X_1+Z_{R})-h(Z_{R})-h(Z_1)\\
&\overset{(c)}\leq\frac{1}{2}\log(1+(h_{11}^2+h_{R1}^2)P)
\end{align*}
(a) and (b) hold since $Z_{R}$ and $Z_1$ are independent of $X_1,X_2,\textbf{X}_{R}$ and (c) holds since $h(h_{11}X_1+Z_1,h_{R1}X_1+Z_{R}) \leq \frac{1}{2}\log(2\pi e)^2 \left(1+h_{11}^2P+h_{R1}^2P\right)$.

Similarly, an upper bound on $R_2$ can be expressed as follows:
\begin{align*}
R_2 \leq \frac{1}{2}\log\left(1+\left(h_{22}^2+h_{R2}^2\right)P\right)
\end{align*}

Thus, the outer bound on the Gaussian interference relay-without-delay channel can be expressed as follows:
\begin{align*}
&R_1 \leq \frac{1}{2}\log(1+(h_{11}^2+h_{R1}^2)P)\\
&R_2 \leq \frac{1}{2}\log(1+(h_{22}^2+h_{R2}^2)P)
\end{align*}

Since the achievable rate region and the outer bound coincide, the capacity region of the very strong Gaussian interference relay-without-delay channel satisfying (\ref{eq:pr}) is given as follows:
\begin{align*}
&R_1 \leq \frac{1}{2}\log(1+(h_{11}^2+h_{R1}^2)P)\\
&R_2 \leq \frac{1}{2}\log(1+(h_{22}^2+h_{R2}^2)P)
\end{align*}
\endproof
In the following, we give an example of a Gaussian interference relay-without-delay channel satisfying the very strong interference condition.

\begin{example}
Consider a Gaussian interference relay-without-delay channel with channel coefficients $h_{11}=h_{22}=h_{R1}=h_{R2}=1, h_{12}=h_{21}=5, h_{2R}^{[1]}=h_{1R}^{[2]}=2, h_{1R}^{[1]}=h_{2R}^{[2]}=1 $ and $P = 1$. For this Gaussian interference relay-without-delay channel, very strong Gaussian interference relay-without-delay channel condition is satisfied since
\begin{align*}
& \left(h_{11}^2+h_{R1}^2\right)P = \left(h_{22}^2+h_{R2}^2\right)P = 2\\
& \frac{\left({h_{11}h_{12}+h_{R1}h_{R2}}\right)^2P}{\left(h_{11}^2+h_{R1}^2\right)^2P+h_{11}^2+h_{R1}^2}= \frac{\left({h_{21}h_{22}+h_{R1}h_{R2}}\right)^2P}{\left(h_{22}^2+h_{R2}^2\right)^2P+h_{22}^2+h_{R2}^2}=6
\end{align*}
For this channel, (\ref{eq:pr}) can be expressed as follows:
\begin{align*}
&\frac{(h_{2R}^{[1]}h_{22}h_{R1}-h_{1R}^{[1]}h_{11}h_{R2})^2+(h_{2R}^{[2]}h_{22}h_{R1}-h_{1R}^{[2]}h_{11}h_{R2})^2}{(h_{1R}^{[2]}h_{2R}^{[1]}-h_{1R}^{[1]}h_{2R}^{[2]})^2 h_{11}^2 h_{22}^2 }\left((h_{R1}^2+h_{R2}^2)P+1\right)=\frac{2}{3} \leq P_R
\end{align*}
For $\frac{2}{3} \leq P_R$, the capacity is given as follows:
\begin{align*}
&R_1 \leq \frac{1}{2}\log3, \\
&R_2 \leq \frac{1}{2}\log3
\end{align*}
\end{example}

\subsection{Strong Gaussian interference relay-without-delay channel}
\begin{definition}
A discrete memoryless interference relay-without-delay channel is said to be strong interference relay-without-delay channel if
\begin{align*}
I(X_1;Y_1|X_2,Y_R,X_R)&\leq I(X_1;Y_2|X_2,Y_R,X_R),\\
I(X_2;Y_2|X_1,Y_R,X_R)&\leq I(X_2;Y_1|X_1,Y_R,X_R)
\end{align*}
for all $p(x_1)p(x_2)p(y_{R}|x_{1},x_{2})p(x_{R}|y_{R})p(y_1,y_2|x_1,x_2,x_R,y_R)$.
\end{definition}
\begin{lemma}
For the Gaussian interference relay-without-delay channel, the strong interference relay-without-delay channel condition is equivalent to $|h_{11}| \leq |h_{21}|$ and $|h_{22}| \leq |h_{12}|$.
\end{lemma}
\begin{proof}
$I(X_1;Y_1|X_2,Y_R,\textbf{X}_R)$ can be expressed as follows:
\begin{align*}
&I(X_1;Y_1|X_2,Y_R,\textbf{X}_R)\\
&=h(Y_1|X_2,Y_R,\textbf{X}_R)-h(Y_1|X_1,X_2,Y_R,\textbf{X}_R)\\
&=h(h_{11}X_1+h_{12}X_2+\textbf{h}_{1R}\textbf{X}_R+Z_1|X_2,h_{R1}X_1+h_{R2}X_2+Z_R,\textbf{X}_R)-h(Z_1)\\
&=h(h_{11}X_1+h_{12}X_2+Z_1|X_2,h_{R1}X_1+Z_R,\textbf{X}_R)-h(Z_1)\\
&\overset{(a)}=h(h_{11}X_1+h_{12}X_2+Z_1|X_2,h_{R1}X_1+Z_R)-h(Z_1)\\
&\overset{(b)}=h(h_{11}X_1+Z_1|h_{R1}X_1+Z_R)-h(Z_1)
\end{align*}
where (a) holds since $h_{11}X_1+h_{12}X_2+Z_1 \rightarrow (X_2,h_{R1}X_1+Z_R) \rightarrow X_R$ and (b) holds since $X_2$ is independent of $(X_1,Z_R,Z_1)$.

Similarly, $I(X_1;Y_2|X_2,Y_R,X_R)$, $I(X_2;Y_2|X_1,Y_R,X_R)$, and $I(X_2;Y_1|X_1,Y_R,X_R)$ can be expressed as follows:
\begin{align*}
I(X_1;Y_2|X_2,Y_R,X_R)&=h(h_{21}X_1+Z_2|h_{R1}X_1+Z_R)-h(Z_2)\\
I(X_2;Y_2|X_1,Y_R,X_R)&=h(h_{22}X_2+Z_2|h_{R2}X_2+Z_R)-h(Z_2)\\
I(X_2;Y_1|X_1,Y_R,X_R)&=h(h_{12}X_2+Z_1|h_{R2}X_2+Z_R)-h(Z_1)
\end{align*}
If $|h_{11}| \leq |h_{21}|$ and $|h_{22}| \leq |h_{12}|$, then the Gaussian BC's $X_1$ to $[(h_{21}X_1+Z_2,h_{R1}X_1+Z_R),(h_{11}X_1+Z_1,h_{R1}X_1+Z_R)]$ and $X_2$ to $[(h_{12}X_2+Z_1,h_{R2}X_2+Z_R),(h_{22}X_2+Z_2,h_{R2}X_2+Z_R)]$ are both degraded, and hence are more capable. Thus, we get $I(X_1;Y_2|X_2,Y_R,X_R) \geq I(X_1;Y_1|X_2,Y_R,X_R)$ and $I(X_2;Y_1|X_1,Y_R,X_R) \geq I(X_2;Y_2|X_1,Y_R,X_R)$ for all $p(x_1)p(x_2)p(x_R|y_R)$. Thus, we get $I(X_1;h_{21}X_1+Z_2,h_{R1}X_1+Z_R) \geq I(X_1;h_{11}X_1+Z_1,h_{R1}X_1+Z_R)$ and $I(X_2;h_{12}X_2+Z_1,h_{R2}X_2+Z_R) \geq I(X_2;h_{22}X_2+Z_2,h_{R2}X_2+Z_R)$ for all $p(x_1)p(x_2)p(x_R|y_R)$. To prove the other direction, assume that $h(h_{11}X_1+Z_1|h_{R1}X_1+Z_R) \leq h(h_{21}X_1+Z_2|h_{R1}X_1+Z_R)$ and $h(h_{22}X_2+Z_2|h_{R2}X_2+Z_R) \leq h(h_{12}X_2+Z_1|h_{R2}X_2+Z_R)$. Assuming $X_1 \sim N(0,P)$ and $X_2 \sim N(0,P)$, we get $|h_{11}| \leq |h_{21}|$ and $|h_{22}| \leq |h_{12}|$.
\end{proof}

\begin{lemma}
For the Gaussian interference relay-without-delay channel with strong interference,
\begin{align*}
I(X_1^n;Y_1^n|X_2^n,Y_R^n)&\leq I(X_1^n;Y_2^n|X_2^n,Y_R^n),\\
I(X_2^n;Y_2^n|X_1^n,Y_R^n)&\leq I(X_2^n;Y_1^n|X_1^n,Y_R^n)
\end{align*}
for all $p(x_1^n)p(x_2^n)$ and $x_{R,j}=f_j(y_{R}^j)$ for all $j \in \{1,2,...n\}$ for all $n \geq 1$.
\end{lemma}

\begin{proof}
%Let $\textbf{X}_R^n=\left(\textbf{X}_{R1}^n,\textbf{X}_{R2}^n,\cdots,\textbf{X}_{RK}^n\right)$ and $\textbf{Y}_R^n=(Y_{R1}^n,Y_{R2}^n,\cdots,Y_{RK}^n)$.
For $n=1$, Lemma 2 holds since
\begin{align*}
&I(X_{1,1};Y_{2,1}|X_{2,1},Y_{R,1}) - I(X_{1,1};Y_{1,1}|X_{2,1},Y_{R,1})\\
&\overset{(a)}=I(X_{1,1};Y_{2,1}|X_{2,1},Y_{R,1},X_{R,1}) - I(X_{1,1};Y_{1,1}|X_{2,1},Y_{R,1},X_{R,1})\\
&\overset{(b)} \geq 0,
\end{align*}
where (a) holds since $X_{R,1}=f_1(Y_{R,1})$ and (b) holds from the strong interference condition.

Now assume that Lemma 2 holds for $n=k$, i.e.,
\begin{align*}
I(X_1^k;Y_1^k|X_2^k,Y_R^k)&\leq I(X_1^k;Y_2^k|X_2^k,Y_R^k)
\end{align*}
for all $p(x_1^k)p(x_2^k)$ and $x_{R,j}=f_j(y_{R}^j)$ for all $j \in \{1,2,...n\}$.

Let $\hat{Y}_{i,j}=h_{i1,j}X_{1,j}+h_{i2,j}X_{2,j}+Z_{i,j}$ and $\hat{Y}_i^k=(\hat{Y}_{i,1},\hat{Y}_{i,2},...,\hat{Y}_{i,k})$ where $i \in \{1,2\}$, $j \in \{1,2,...,n\}$.
Then, the above condition is equivalent to the following:
\begin{align*}
I(X_1^k;\hat{Y}_1^k|X_2^k,Y_R^k)&\leq I(X_1^k;\hat{Y}_2^k|X_2^k,Y_R^k)
\end{align*}
for all $p(x_1^k)p(x_2^k)$ since $X_{R}^k=f^k(Y_{R}^k)$. This implies that
\begin{align*}
I(X_1^k;\hat{Y}_1^k|X_2^k,Y_R^k,U)&\leq I(X_1^k;\hat{Y}_2^k|X_2^k,Y_R^k,U)
\end{align*}
for all $p(u)p(x_1^k|u)p(x_2^k|u)$.
%and Gaussian channel outputs $Y_R^k,Y_1^k,Y_2^k$ satisfying $\prod_{j=1}^{k}p(y_{R,j}|x_{1,j},x_{2,j})p(\hat{y}_{1,j}|x_{1,j},x_{2,j})p(\hat{y}_{2,j}|x_{1,j},x_{2,j})$.

Then $I(X_1^{k+1};Y_2^{k+1}|X_2^{k+1},Y_R^{k+1}) - I(X_1^{k+1};Y_1^{k+1}|X_2^{k+1},Y_R^{k+1})$ can be expressed as follows:
\begin{align*}
&I(X_1^{k+1};Y_2^{k+1}|X_2^{k+1},Y_R^{k+1})-I(X_1^{k+1};Y_1^{k+1}|X_2^{k+1},Y_R^{k+1})\\
&\overset{(a)}=I(X_1^{k+1};\hat{Y}_2^{k+1}|X_2^{k+1},Y_R^{k+1})-I(X_1^{k+1};\hat{Y}_1^{k+1}|X_2^{k+1},Y_R^{k+1})\\
&=I(X_1^{k+1};\hat{Y}_2^{k}|X_2^{k+1},Y_R^{k+1})+I(X_1^{k+1};\hat{Y}_{2,k+1}|X_2^{k+1},Y_R^{k+1},\hat{Y}_2^{k})\\
&-I(X_1^{k+1};\hat{Y}_{1,k+1}|X_2^{k+1},Y_R^{k+1})-I(X_1^{k+1};\hat{Y}_1^{k}|X_2^{k+1},Y_R^{k+1},\hat{Y}_{1,k+1})\\
&=I(X_1^{k+1},\hat{Y}_{1,k+1};\hat{Y}_2^{k}|X_2^{k+1},Y_R^{k+1})+I(X_1^{k+1};\hat{Y}_{2,k+1}|X_2^{k+1},Y_R^{k+1},\hat{Y}_2^{k})\\
&-I(X_1^{k+1},\hat{Y}_2^{k};\hat{Y}_{1,k+1}|X_2^{k+1},Y_R^{k+1})-I(X_1^{k+1};\hat{Y}_1^{k}|X_2^{k+1},Y_R^{k+1},\hat{Y}_{1,k+1})\\
&=I(\hat{Y}_{1,k+1};\hat{Y}_2^{k}|X_2^{k+1},Y_R^{k+1})+I(X_1^{k};\hat{Y}_2^{k}|X_2^{k+1},Y_R^{k+1},\hat{Y}_{1,k+1})\\
&+I(X_{1,k+1};\hat{Y}_2^{k}|X_1^{k},X_2^{k+1},Y_R^{k+1},\hat{Y}_{1,k+1})+I(X_{1,k+1};\hat{Y}_{2,k+1}|X_2^{k+1},Y_R^{k+1},\hat{Y}_2^{k})\\
&+I(X_{1}^k;\hat{Y}_{2,k+1}|X_{1,k+1},X_2^{k+1},Y_R^{k+1},\hat{Y}_2^{k})-I(\hat{Y}_2^{k};\hat{Y}_{1,k+1}|X_2^{k+1},Y_R^{k+1})\\
&-I(X_{1,k+1};\hat{Y}_{1,k+1}|X_2^{k+1},Y_R^{k+1},\hat{Y}_2^{k})-I(X_{1}^k;\hat{Y}_{1,k+1}|X_{1,k+1},X_2^{k+1},Y_R^{k+1},\hat{Y}_2^{k})\\
&-I(X_{1}^k;\hat{Y}_1^{k}|X_2^{k+1},Y_R^{k+1},\hat{Y}_{1,k+1})-I(X_{1,k+1};\hat{Y}_1^{k}|X_1^k,X_2^{k+1},Y_R^{k+1},\hat{Y}_{1,k+1})\\
&\overset{(b)}=I(X_1^{k};\hat{Y}_2^{k}|X_2^{k+1},Y_R^{k+1},\hat{Y}_{1,k+1})+I(X_{1,k+1};\hat{Y}_{2,k+1}|X_2^{k+1},Y_R^{k+1},\hat{Y}_2^{k})\\
&-I(X_{1,k+1};\hat{Y}_{1,k+1}|X_2^{k+1},Y_R^{k+1},\hat{Y}_2^{k})-I(X_{1}^k;\hat{Y}_1^{k}|X_2^{k+1},Y_R^{k+1},\hat{Y}_{1,k+1})\\
&\overset{(c)}\geq I(X_{1,k+1};\hat{Y}_{2,k+1}|X_2^{k+1},Y_R^{k+1},\hat{Y}_2^{k})-I(X_{1,k+1};\hat{Y}_{1,k+1}|X_2^{k+1},Y_R^{k+1},\hat{Y}_2^{k})\\
&\overset{(d)}\geq 0
\end{align*}
where (a) holds since $X_{R}^k=f^k(Y_{R}^k)$ and translation property of differential entropy, (b) holds since $(X_{1,k+1},X_{2,k+1},Y_{R}^{k+1},\hat{Y}_{1,k+1}) \rightarrow (X_1^k,X_2^k) \rightarrow \hat{Y}_2^k$, $(X_{1}^k,X_{2}^k,Y_{R}^{k+1},\hat{Y}_{2,k}) \rightarrow (X_{1,k+1},X_{2,k+1})$ $\rightarrow \hat{Y}_{2,k+1}$, $(X_{1}^k,X_{2}^k,Y_{R}^{k+1},\hat{Y}_{2,k}) \rightarrow (X_{1,k+1},X_{2,k+1}) \rightarrow \hat{Y}_{1,k+1}$, and $(X_{1,k+1},X_{2,k+1},Y_{R}^{k+1},\hat{Y}_{1,k+1})$ $\rightarrow (X_{1}^k,X_{2}^k) \rightarrow \hat{Y}_{1}^k$, (c) follows from the assumption since $(X_{2,k+1},Y_{R,k+1},\hat{Y}_{1,k+1}) \rightarrow (X_2^k,X_1^k) \rightarrow (\hat{Y}_1^k,\hat{Y}_2^k,Y_{R}^k)$, $X_1^k \rightarrow (X_{2,k+1},Y_{R,k+1},\hat{Y}_{1,k+1}) \rightarrow X_2^k$, and the memoryless channel property, and (d) follows from the strong interference condition since $(X_{2}^{k},Y_{R}^{k},\hat{Y}_{2}^{k}) \rightarrow (X_{1,k+1},X_{2,k+1}) \rightarrow (\hat{Y}_{1,k+1},\hat{Y}_{2,k+1},Y_{R,k+1})$ and $X_{1,k+1} \rightarrow (X_{2}^{k},Y_{R}^{k},\hat{Y}_{2}^{k}) \rightarrow X_{2,k+1}$.
Similarly, we can prove that $I(X_2^n;Y_2^n|X_1^n,Y_R^n) \leq I(X_2^n;Y_1^n|X_1^n,Y_R^n)$.
\end{proof}

Using Lemma 2, we get the following outer bound on the capacity region for the strong Gaussian interference relay-without-delay channel.
\begin{theorem}
For the strong Gaussian interference relay-without-delay channel, the capacity region is contained in the set of rate pairs $(R_1,R_2)$ such that
\begin{align*}
%R_1 &\leq I(X_{1};Y_{R}|X_{2},Q)+I(X_{1};Y_{1}|X_{2},Y_{R},X_{R},Q),\\
%R_2 &\leq I(X_{2};Y_{R}|X_{1},Q)+I(X_{2};Y_{2}|X_{1},Y_{R},X_{R},Q),\\
%R_1+R_2 &\leq \min \left(I(X_{1},X_{2};Y_{R}|Q)+I(X_{1},X_{2};Y_{2}|Y_{R},X_{R},Q),I(X_{1},X_{2};Y_{R}|Q)+I(X_{1},X_{2};Y_{1}|Y_{R},X_{R},Q)\right)
R_1 &\leq \frac{1}{2}\log\left(1+(h_{11}^2+h_{R1}^2)P\right)\\
R_2 &\leq \frac{1}{2}\log\left(1+(h_{22}^2+h_{R2}^2)P\right)\\
R_1+R_2 &\leq \frac{1}{2}\log\left(1+\left(h_{11}^2+h_{12}^2+h_{R1}^2+h_{R2}^2\right)P+\left((h_{11}^2+h_{12}^2)(h_{R1}^2+h_{R2}^2)-(h_{11}h_{R1}+h_{12}h_{R2})^2\right)P^2\right)\\
R_1+R_2 &\leq \frac{1}{2}\log\left(1+\left(h_{21}^2+h_{22}^2+h_{R1}^2+h_{R2}^2\right)P+\left((h_{21}^2+h_{22}^2)(h_{R1}^2+h_{R2}^2)-(h_{21}h_{R1}+h_{22}h_{R2})^2\right)P^2\right)
\end{align*}
\end{theorem}

\begin{proof}
From Fano's inequality \cite{text_Cover}, we get
\begin{align*}
&H(M_1|Y_1^n,Y_{R}^n,M_2)\leq H(M_1|Y_1^n)\leq H(M_1|\hat{M_1})\leq n\epsilon_n,\\
&H(M_2|Y_2^n,Y_{R}^n,M_1)\leq H(M_2|Y_2^n)\leq H(M_2|\hat{M_2})\leq n\epsilon_n,\\
&H(M_1,M_2|Y_1^n,Y_2^n,Y_{R}^n)\leq H(M_1,M_2|Y_1^n,Y_2^n)\leq H(M_1,M_2|\hat{M_1},\hat{M_2})\leq n\epsilon_n.
\end{align*}
where $\epsilon_n$ tends to zero as $n \rightarrow \infty$.
An upper bound on $nR_1+nR_2$ can be expressed as follows:
\begin{align*}
nR_1+nR_2&=H(M_1)+H(M_2)\\
&\overset{(a)}\leq I(X_1^n;Y_1^n,Y_{R}^n)+I(X_2^n;Y_2^n,Y_{R}^n)+2n\epsilon_n\\
&\leq I(X_1^n;Y_1^n,Y_{R}^n|X_2^n)+I(X_2^n;Y_2^n,Y_{R}^n)+2n\epsilon_n\\
&= I(X_1^n;Y_{R}^n|X_2^n)+I(X_1^n;Y_1^n|Y_{R}^n,X_2^n)+I(X_2^n;Y_{R}^n)+I(X_2^n;Y_{2}^n|Y_{R}^n)+2n\epsilon_n\\
&= I(X_1^n,X_2^n;Y_{R}^n)+I(X_1^n;Y_1^n|Y_{R}^n,X_2^n)+I(X_2^n;Y_{2}^n|Y_{R}^n)+2n\epsilon_n\\
&\overset{(b)}\leq I(X_1^n,X_2^n;Y_{R}^n)+I(X_1^n;Y_2^n|Y_{R}^n,X_2^n)+I(X_2^n;Y_{2}^n|Y_{R}^n)+2n\epsilon_n\\
&=I(X_1^n,X_2^n;Y_{R}^n)+I(X_1^n,X_2^n;Y_2^n|Y_{R}^n)+2n\epsilon_n\\
&=\sum_{i=1}^nI(X_1^n,X_2^n;Y_{R,i}|Y_{R}^{i-1})+\sum_{i=1}^nI(X_1^n,X_2^n;Y_{2,i}|Y_{2}^{i-1},Y_{R}^n)+2n\epsilon_n\\
&\overset{(c)}\leq\sum_{i=1}^nI(X_1^n,X_2^n,Y_{R}^{i-1};Y_{R,i})+\sum_{i=1}^nI(X_1^n,X_2^n,Y_{2}^{i-1},Y_{R}^{i-1},Y_{R,i+1}^{n};Y_{2i}|Y_{R,i},X_{R,i})+2n\epsilon_n\\
&\overset{(d)}=\sum_{i=1}^nI(X_{1i},X_{2i};Y_{R,i})+\sum_{i=1}^nI(X_{1i},X_{2i};Y_{2i}|Y_{R,i},X_{R,i})+2n\epsilon_n\\
&=nI(X_{1Q},X_{2Q};Y_{R,Q}|Q)+nI(X_{1Q},X_{2Q};Y_{2Q}|Y_{R,Q},X_{R,Q},Q)+2n\epsilon_n\\
&=nI(X_{1},X_{2};Y_{R}|Q)+nI(X_{1},X_{2};Y_{2}|Y_{R},X_{R},Q)+2n\epsilon_n
\end{align*}
where (a) follows from Fano's inequality, (b) follows from Lemma 1, (c) holds since $X_{R,i}=f_i(Y_{R}^{i})$, and (d) holds since $(X_1^{i-1},X_{1,i+1}^{n},X_2^{i-1},X_{2,i+1}^{n},Y_2^{i-1},Y_{R}^{i-1},Y_{R,i+1}^{n}) \rightarrow (X_{1i},X_{2i},Y_{R,i},X_{R,i})$ $\rightarrow Y_{2i}$ and $(X_1^{i-1},X_{1,i+1}^{n},X_2^{i-1},X_{2,i+1}^{n},Y_{R}^{i-1}) \rightarrow (X_{1i},X_{2i}) \rightarrow Y_{R,i}$. In the above, $Q$ is a time sharing random variable such that $Q \sim \operatorname{Unif}[1:n]$ and independent of $(X_1^n,X_2^n,Y_1^n,Y_2^n,Y_R^n,X_R^n)$ and we define $X_1=X_{1Q}$, $X_2=X_{2Q}$, $Y_1=Y_{1Q}$, $Y_2=Y_{2Q}$, $Y_R=Y_{RQ}$, and $X_R=X_{RQ}$. Similarly, we get the following upper bound on $R_1+R_2$.
\begin{align*}
R_1+R_2&\leq I(X_{1},X_{2};Y_{R}|Q)+I(X_{1},X_{2};Y_{1}|Y_{R},X_{R},Q)
\end{align*}
Combining this result with Theorem 1, we get the following outer bound for the strong Gaussian interference relay-without-delay channel
\begin{align*}
R_1 &\leq I(X_{1};Y_{R}|X_{2},Q)+I(X_{1};Y_{1}|X_{2},Y_{R},X_{R},Q),\\
R_2 &\leq I(X_{2};Y_{R}|X_{1},Q)+I(X_{2};Y_{2}|X_{1},Y_{R},X_{R},Q),\\
R_1+R_2 &\leq I(X_{1},X_{2};Y_{R}|Q)+ \min \{I(X_{1},X_{2};Y_{2}|Y_{R},X_{R},Q),I(X_{1},X_{2};Y_{1}|Y_{R},X_{R},Q)\}
\end{align*}
for all $p(q)p(x_1|q)p(x_2|q)p(x_{R}|y_{R},q)$.
Following the same procedure in the converse proof in subsection IV. A, we get the following upper bounds on $R_1$ and $R_2$
\begin{align*}
R_1 &\leq \frac{1}{2}\log\left(1+(h_{11}^2+h_{R1}^2)P\right)\\
R_2 &\leq \frac{1}{2}\log\left(1+(h_{22}^2+h_{R2}^2)P\right).
\end{align*}
For the Gaussian case, the above upper bound on $R_1+R_2$ can be expressed as follows:
\begin{align*}
&R_1+R_2 \leq I(X_{1},X_{2};Y_{R})+I(X_{1},X_{2};Y_{1}|Y_{R},\textbf{X}_{R})\\
&=h(Y_{R})-h(Y_{R}|X_1,X_2)+h(Y_1|Y_{R},\textbf{X}_{R})-h(Y_1|X_1,X_2,Y_{R},\textbf{X}_{R})\\
&\overset{(a)}=h(h_{R1}X_1+h_{R2}X_2+Z_{R})-h(Z_{R})+h(h_{11}X_1+h_{12}X_2+h_{1R}\textbf{X}_{R}+Z_1|Y_{R},\textbf{X}_{R})-h(Z_1)\\
&\overset{(b)}=h(h_{R1}X_1+h_{R2}X_2+Z_{R})-h(Z_{R})+h(h_{11}X_1+h_{12}X_2+Z_1|h_{R1}X_1+h_{R2}X_2+Z_{R})-h(Z_1)\\
&=h(h_{R1}X_1+h_{R2}X_2+Z_{R},h_{11}X_1+h_{12}X_2+Z_1)-h(Z_{R})-h(Z_1)\\
&\leq \frac{1}{2}\log\left(1+\left(h_{11}^2+h_{12}^2+h_{R1}^2+h_{R2}^2\right)P+\left((h_{11}^2+h_{12}^2)(h_{R1}^2+h_{R2}^2)-(h_{11}h_{R1}+h_{12}h_{R2})^2\right)P^2\right)
\end{align*}
where (a) and (b) hold since $Z_{R}$ and $Z_1$ are independent of $X_1$, $X_2$, and $\textbf{X}_{R}$.% and (c) holds since $h(h_{R,1}X_1+h_{R,2}X_2+Z_{R},h_{11}X_1+h_{12}X_2+Z_1)$\\ $\leq\frac{1}{2}\log(2\pi e)^2\left(1+\left(h_{11}^2+h_{12}^2+h_{R,1}^2+h_{R,2}^2\right)P+(h_{11}^2+h_{12}^2)(h_{R,1}^2+h_{R,2}^2)-(h_{11}h_{R,1}+h_{12}h_{R,2})^2\right)P^2$.

Similarly, we get,
\begin{align*}
R_1+R_2 &\leq \frac{1}{2}\log\left(1+\left(h_{21}^2+h_{22}^2+h_{R1}^2+h_{R2}^2\right)P+\left((h_{21}^2+h_{22}^2)(h_{R1}^2+h_{R2}^2)-(h_{21}h_{R1}+h_{22}h_{R2})^2\right)P^2\right),
\end{align*}
which concludes the theorem.
\end{proof}

\begin{theorem}
The capacity region of the strong Gaussian interference relay-without-delay channel for
\begin{align}
&\frac{(h_{2R}^{[1]}h_{22}h_{R1}-h_{1R}^{[1]}h_{11}h_{R2})^2+(h_{2R}^{[2]}h_{22}h_{R1}-h_{1R}^{[2]}h_{11}h_{R2})^2}{(h_{1R}^{[2]}h_{2R}^{[1]}-h_{1R}^{[1]}h_{2R}^{[2]})^2 h_{11}^2 h_{22}^2 } \leq \frac{P_R}{(h_{R1}^2+h_{R2}^2)P+1},\label{eq:pr2} \\
&(h_{11},h_{12})=a_0(h_{21},h_{22})=a_1(h_{R1},h_{R2}) \nonumber
\end{align}
is the set of rate pairs $(R_1,R_2)$ such that
\begin{align*}
R_1 &\leq \frac{1}{2}\log\left(1+(h_{11}^2+h_{R1}^2)P\right)\\
R_2 &\leq \frac{1}{2}\log\left(1+(h_{22}^2+h_{R2}^2)P\right)\\
R_1+R_2 &\leq \frac{1}{2}\log\left(1+\left(h_{11}^2+h_{12}^2+h_{R1}^2+h_{R2}^2\right)P\right)\\
R_1+R_2 &\leq \frac{1}{2}\log\left(1+\left(h_{22}^2+h_{21}^2+h_{R1}^2+h_{R2}^2\right)P\right)
\end{align*}
where $a_0$ and $a_1$ are any real numbers.
\end{theorem}

\begin{proof}

\textbf{Achievability.} Achievability follows by using simultaneous non-unique decoding and instantaneous amplify-and-forward relaying scheme described in the very strong interference channel. From the simultaneous non-unique decoding, the achievable rate pairs $(R_1,R_2)$ can be expressed as follows:
\begin{align*}
&R_1 \leq \frac{1}{2}\log(1+(h_{11}^2+h_{R1}^2)P)\\
&R_2 \leq \frac{1}{2}\log(1+(h_{22}^2+h_{R2}^2)P)\\
&R_1+R_2 \leq \frac{1}{2}\log\left(1+\left(h_{11}^2+h_{R1}^2+\frac{(h_{11}h_{12}+h_{R1}h_{R2})^2}{h_{11}^2+h_{R1}^2}\right)P\right)\\%C
&R_1+R_2 \leq \frac{1}{2}\log\left(1+\left(h_{22}^2+h_{R2}^2+\frac{(h_{21}h_{22}+h_{R1}h_{R2})^2}{h_{22}^2+h_{R2}^2}\right)P\right)%D
\end{align*}
From the assumption, we get
\begin{align*}
(h_{11}h_{12}+h_{R1}h_{R2})^2 = (h_{11}^2+h_{R1}^2)(h_{12}^2+h_{R2}^2)\\
(h_{21}h_{22}+h_{R1}h_{R2})^2 = (h_{21}^2+h_{R1}^2)(h_{22}^2+h_{R2}^2)
\end{align*}
since $ (h_{11},h_{12})=a_0(h_{21},h_{22})=a_1(h_{R1},h_{R2}) $ implies $ (h_{11},h_{R1})=c(h_{12},h_{R2}), (h_{21},h_{R11})=c(h_{22},h_{R2})$, where $c$ is a constant.
Thus, the achievable rate pairs $(R_1,R_2)$ can be expressed as follows:
\begin{align*}
&R_1 \leq \frac{1}{2}\log(1+(h_{11}^2+h_{R1}^2)P)\\
&R_2 \leq \frac{1}{2}\log(1+(h_{22}^2+h_{R2}^2)P)\\
&R_1+R_2 \leq \frac{1}{2}\log\left(1+\left(h_{11}^2+h_{12}^2+h_{R1}^2+h_{R2}^2\right)P\right)\\
&R_1+R_2 \leq \frac{1}{2}\log\left(1+\left(h_{21}^2+h_{22}^2+h_{R1}^2+h_{R2}^2\right)P\right)
\end{align*}

\textbf{Converse.}  From Theorem 3, we get the following outer bound on $(R_1,R_2)$
\begin{align*}
&R_1 \leq \frac{1}{2}\log(1+(h_{11}^2+h_{R1}^2)P)\\
&R_2 \leq \frac{1}{2}\log(1+(h_{22}^2+h_{R2}^2)P)\\
&R_1+R_2 \leq \frac{1}{2}\log\left(1+\left(h_{11}^2+h_{12}^2+h_{R1}^2+h_{R2}^2\right)P+\left((h_{11}^2+h_{12}^2)(h_{R1}^2+h_{R2}^2)-(h_{11}h_{R1}+h_{12}h_{R2})^2\right)P^2\right)\\
&R_1+R_2 \leq \frac{1}{2}\log\left(1+\left(h_{21}^2+h_{22}^2+h_{R1}^2+h_{R2}^2\right)P+\left((h_{21}^2+h_{22}^2)(h_{R1}^2+h_{R2}^2)-(h_{21}h_{R1}+h_{22}h_{R2})^2\right)P^2\right)
\end{align*}
From the assumption, we get the following equation:
\begin{align*}
(h_{11}h_{R1}+h_{12}h_{R2})^2 &= (h_{11}^2+h_{12}^2)(h_{R1}^2+h_{R2}^2),\\ %\hspace{0.1in} \forall i,j \in \{1,R1,...,RK\}, i \neq j
(h_{21}h_{R1}+h_{22}h_{R2})^2 &= (h_{21}^2+h_{22}^2)(h_{R1}^2+h_{R2}^2)
\end{align*}

Using this, we get:
\begin{align*}
&R_1 \leq \frac{1}{2}\log(1+(h_{11}^2+h_{R1}^2)P)\\
&R_2 \leq \frac{1}{2}\log(1+(h_{22}^2+h_{R2}^2)P)\\
&R_1+R_2 \leq \frac{1}{2}\log\left(1+\left(h_{11}^2+h_{12}^2+h_{R1}^2+h_{R2}^2\right)\right)\\
&R_1+R_2 \leq \frac{1}{2}\log\left(1+\left(h_{22}^2+h_{21}^2+h_{R1}^2+h_{R2}^2\right)\right)
\end{align*}
\end{proof}

In the following, we will give an example of a Gaussian interference relay-without-delay channel satisfying the strong interference condition.

\begin{example}
Consider a Gaussian interference relay-without-delay channel with channel coefficients $h_{11}=h_{21}=h_{R1}=1$, $h_{12}=h_{22}=h_{R2}=2$, $h_{1R}^{[1]}=1,h_{1R}^{[2]}=3,h_{2R}^{[1]}=3,h_{2R}^{[2]}=3,$ and $P = 1$. For this Gaussian interference relay-without-delay channel, the strong interference Gaussian relay-without-delay channel condition is satisfied since
\begin{align*}
1 = |h_{11}| &\leq |h_{21}| = 1\\
2 = |h_{22}| &\leq |h_{12}| = 2.
\end{align*}
Furthermore, we have
\begin{align*}
(h_{11},h_{12})=(h_{21},h_{22})=(h_{R1},h_{R2})&=(1,2).
\end{align*}
For this channel, (\ref{eq:pr2}) can be expressed as follows:
\begin{align*}
&\frac{(h_{2R}^{[1]}h_{22}h_{R1}-h_{1R}^{[1]}h_{11}h_{R2})^2+(h_{2R}^{[2]}h_{22}h_{R1}-h_{1R}^{[2]}h_{11}h_{R2})^2}{(h_{1R}^{[2]}h_{2R}^{[1]}-h_{1R}^{[1]}h_{2R}^{[2]})^2 h_{11}^2 h_{22}^2 }\left((h_{R1}^2+h_{R2}^2)P+1\right)=\frac{2}{3} \leq P_R
\end{align*}
For $\frac{2}{3} \leq P_R$, the capacity is given as follows:
\begin{align*}
&R_1 \leq \frac{1}{2}\log3, \\
&R_2 \leq \frac{1}{2}\log9, \\
&R_1+R_2 \leq \frac{1}{2}\log11 \\
\end{align*}
\end{example}

\section{Conclusion}

In this paper, we studied the interference relay-without-delay channel.
The main difference between the interference relay-without-delay channel and the conventional interference relay channel is in that the relay's current transmit symbol depends on its current received symbol as well as on the past received sequences.
For this channel, we defined very strong and strong Gaussian interference relay-without-delay channel motivated by the strong and very strong Gaussian interference channels.
We presented new outer bounds using genie-aided outer bounding such that both receivers know the received sequences of the relay.
Using the characteristics of the strong interference relay-without-delay channel, we suggested a tighter outer bound for the strong Gaussian interference relay-without-delay channel.
We also proposed an achievable scheme based on instantaneous amplify-and-forward relaying.
Surprisingly, the proposed instantaneous amplify-and-forward relaying scheme actually achieves the capacity exactly under certain conditions for both very strong and strong Gaussian interference relay-without-delay channels despite its simplicity.
 The proposed scheme would be practically useful since it can be optimal and at the same time it is very simple.
%Thus, it is practically useful as well and can be easily applied in many practical scenarios.
%Thus, it is important not only academically but also practically and is applicable in many practical scenarios because of its simplicity.
%Thus, we argue that the best strategy of relay for certain network scenario is simple instantaneous amplify-and-forward relaying.

\end{document}